\documentclass[journal,comsoc]{IEEEtran}
\usepackage[T1]{fontenc}% optional T1 font encoding
\usepackage{graphicx}
\usepackage{epstopdf}
\usepackage{caption}
\usepackage{subcaption}
\usepackage{amsthm}
\usepackage{bm}
\usepackage{tabularx}
\usepackage[utf8]{inputenc}
\usepackage[english]{babel}
\newtheorem{theorem}{Theorem}

\ifCLASSINFOpdf
\else
\fi
\usepackage[inline]{enumitem}
\usepackage{amsmath}
\usepackage[cmintegrals]{newtxmath}
\begin{document}
\title{Constrained Receiver Scheduling in Flexible Time and Wavelength Division Multiplexed Optical Access Networks}

\author{Chayan~Bhar, Arnab Mitra, Goutam~Das% <-this % stops a space
\thanks{C. Bhar is a PhD scholar in the G.S.Sanyal School of Telecommunications, Indian Institute of Technology Kharagpur, India. e-mail: (chayanbhar88@live.com).}% <-this % stops a space
\thanks{Goutam Das is Asst. Professor in the G.S.Sanyal School of Telecommunications, Indian Institute of Technology Kharagpur, India. e-mail: (gdas.gssst.iitkgp.ernet.in).}}%


\markboth{Journal of \LaTeX\ Class Files,~Vol.~14, No.~8, August~2015}%
{Shell \MakeLowercase{\textit{et al.}}: Bare Demo of IEEEtran.cls for IEEE Communications Society Journals}



\maketitle

\begin{abstract}
An increasing bandwidth demand has mandated a shift to the time and wavelength division multiplexing (TWDM) techniques in optical access networks (OAN). Typical TWDM scheduling schemes consider scheduling of the optical line terminal receiver only. In this paper we have identified an additional collision domain that is present in TWDM schemes that offer security, in addition to bandwidth flexibility. Scheduling of the identified collision domain is termed as group scheduling. We illustrate that consideration of receiver scheduling only (as done in typical TWDM schemes) severely affects their throughput when implemented on flexible and secure TWDM architectures. A novel media access control protocol has been proposed in this paper that considers the multiple collision domains. Through simulations, we are able to illustrate that the proposed scheme achieves a high throughput.  A theoretical upper bound of throughput has also been derived to explain the simulation results. Complexity reduction of the proposed scheme has been illustrated, thereby making it an attractive proposal.
\end{abstract}


% \fi
\IEEEpeerreviewmaketitle

\section{Introduction} \label{Section:Intro}

\IEEEPARstart{T}{he} introduction of bandwidth intensive and quality of service (QoS) aware internet applications has resulted in a high per-user bandwidth demand. Optical access networks (OANs) were proposed to facilitate high end-user bandwidth and ensure an excellent QoS. OANs consist of an optical line terminal (OLT) at the central office (Fig. \ref{Fig:TWDM architectures}). The end-user units in an OAN are the optical network units (ONUs). The OLT performs bandwidth allocation among the ONUs through multiple stages of on-field remote nodes (first stage - $RN_1$ and second stage - $RN_{2,x}$). Bandwidth allocation is performed using statistical multiplexing methods. The hybrid time and wavelength division multiplexed (TWDM) scheme has been approved by the full service access network group as the next generation OAN technology \cite{Kani2009}. TWDM OANs enable the OLT to allocate bandwidth to ONUs on multiple wavelengths. This allows such schemes to support a high end-user bandwidth. Moreover, TWDM schemes can be designed so as to meet certain desirable criteria that are discussed below.

\begin{itemize}
	\item Bandwidth flexibility: Allows routing of available  bandwidth to anywhere within the network.
	\item Security and privacy: Prevents crosstalk attacks and unintended reception of data by malicious users in upstream and downstream respectively.
	\item Passivity: Nullifies the need of active on-field routing elements (that require power provisioning and significant operational expenditures).
	\item Excellent reach: Nullifies the need for on-field amplifiers thereby reducing capital expenditures and the need for power provisioning.
\end{itemize}

The TWDM schemes proposed in literature have been designed to meet the above objectives. Below we categorise the existing TWDM schemes and thereafter provide a qualitative comparison of these schemes. The classification has been performed (Fig. \ref{Fig:TWDM architectures}) depending on the OLT and distribution network designs. 

\begin{figure*}
	\centering
	\includegraphics[clip, page=4, trim=0cm 0cm 0cm 0.5cm, width=1\textwidth]{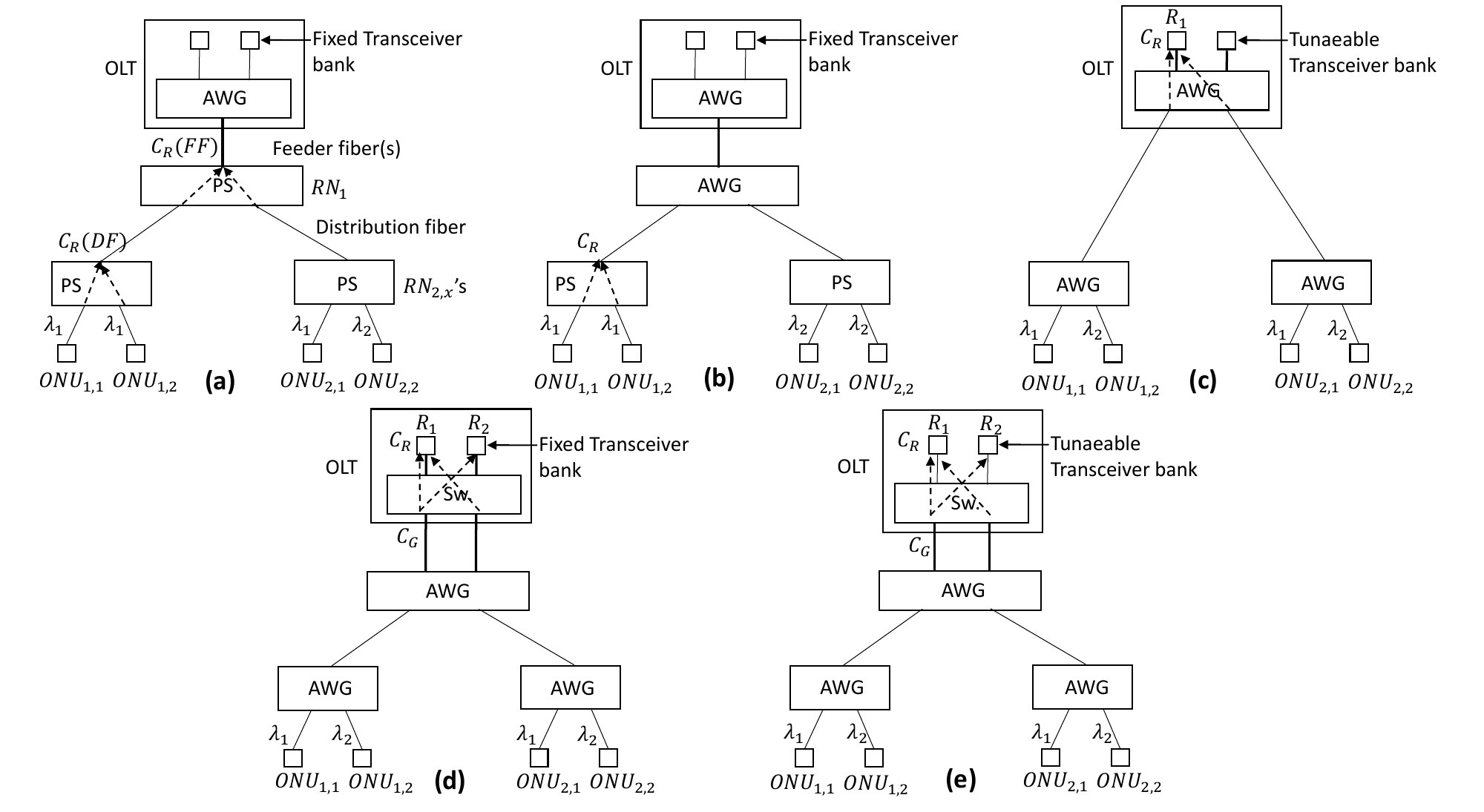}
	\caption{Architectures for different types of TWDM networks proposed in the literature \cite{Luo2013}, \cite{Dixit2011}, \cite{Bock2005}, \cite{Tsalamanis2004}, \cite{Bhar2014}, $G_1-ONU_{1,1}, ONU_{2,2}$, $G_2-ONU_{1,2}, ONU_{2,1}$.}
	\label{Fig:TWDM architectures}
\end{figure*}

\begin{enumerate}
	\item OLT has a fixed transceiver bank connected to an arrayed waveguide grating (AWG) and requires a single feeder fiber (Fig. \ref{Fig:TWDM architectures}(a)) \cite{Luo2013}. The distribution network has power splitter (PS) at $RN_1$ followed by another PS at the second stage remote node ($RN_{2,x}$). Schemes of this type ensure bandwidth flexibility but lack security, privacy and sufficient reach.
	
	\item OLT has a fixed transceiver bank connected to multiplexer and requires a single feeder fiber (Fig. \ref{Fig:TWDM architectures}(b)) \cite{Dixit2011}. The distribution network has an arrayed waveguide grating (AWG) device at $RN_1$, which is followed by PS at $RN_{2,x}$. Schemes of this type have limited bandwidth flexibility, security, privacy and reach.
	
	\item OLT has fixed a transceiver bank connected to a switch and requires multiple feeder fibers (Fig. \ref{Fig:TWDM architectures}(c)) \cite{Bock2005}. The distribution network has an AWG at $RN_1$ followed by AWGs at $RN_{2,x}$. Schemes of this type have limited flexibility but excellent security, privacy and reach.
	
	\item OLT has a tunable laser bank and requires multiple feeder fibers (Fig. \ref{Fig:TWDM architectures}(d)) \cite{Tsalamanis2004}. The distribution network is similar to that of Fig. \ref{Fig:TWDM architectures}(c) and has similar features.
	
	\item OLT has a tunable laser bank followed by two switches, each for upstream and downstream (Fig. \ref{Fig:TWDM architectures}(e)) \cite{Bhar2014}. Multiple feeder fibers are required and the distribution network has an AWG at $RN_1$ followed by another AWG at $RN_{2,x}$. Schemes of this type have excellent flexibility, security, privacy and reach.
\end{enumerate}

The first scheme allows ONUs to have tunable transceivers while the other schemes mandate the ONUs to have fixed wavelength transceivers. The schemes illustrated in Fig. \ref{Fig:TWDM architectures} are associated with collision domains at different points ($P_1,~P_2$) in the network. Therefore each of the schemes has a different scheduling requirement which is essential to prevent collision at these points. Moreover, since downstream is broadcast to the ONUs, scheduling is required to be performed for upstream only. We illustrate the implication of imperfect scheduling, using Fig. \ref{Fig:Collision}. For the scenario of Fig. \ref{Fig:Collision}, we assume that two receivers are present at the OLT. Moreover two $RN_{2,x}$'s and two ONUs per $RN_{2,x}$ are present in the distribution network ($RN_{2,1}:~ONU_{1,1},~ONU_{1,2}$ and $RN_{2,2}:~ONU_{2,1},~ONU_{2,2}$).

\subsection{Scheduling requirements in different TWDM schemes}

In schemes of the first type (Fig. \ref{Fig:TWDM architectures}(a)), upstream collision can occur at $C_R(DF)$ and (or) $C_R(FF)$. This happens if two ONUs, upstream on the same wavelength ($\lambda_1$) at overlapping time intervals. The ONUs can be from the same $RN_{2,x}$ (e.g., $ONU_{1,1}$ and $ONU_{1,2}$) or different $RN_{2,x}$'s (e.g., $ONU_{1,1}$ and $ONU_{2,1}$). In the first case collision occurs at both $C_R^1$ and $C_R^2$ while in the second case, collision is observed at $C_R^2$ only. Data collision is depicted by $C_R$ in Fig. \ref{Fig:Collision} using the upstream data of $ONU_{1,1}$ and $ONU_{2,1}$ ($U_{1,1}$ and $U_{2,1}$ respectively). We term this type of collision as \textit{Receiver collision}. Any collision at $C_R^1$ or $C_R^2$ is reflected at the OLT receivers. If ONUs are properly scheduled to the OLT receivers, such that multiple ONUs never upstream on the same wavelength at overlapping time intervals, then collision can be avoided at $C_R^1$ and $C_R^2$. 

Schemes of the second type (Fig. \ref{Fig:TWDM architectures}(b)) have a collision domain at $RN_{2,x}$ ($C_R$). Such architectures have the limitation that ONUs connected to a particular $RN_{2,x}$ communicate with the OLT on the same wavelength. Therefore, if two ONUs connected to a particular $RN_{2,x}$, e.g., $ONU_{1,1}$ and $ONU_{1,2}$ upstream simultaneously (or at overlapping time intervals) data corruption is inevitable (\textit{Receiver collision} similar to Fig. \ref{Fig:TWDM architectures}(a)). This can be prevented by proper scheduling of ONUs connecting to a particular $RN_{2,x}$.

In TWDM architectures of the third type (Fig. \ref{Fig:TWDM architectures}(c)) collision $C_R$ occurs at $R_1$, if ONUs upstream to the same OLT receiver (on different wavelengths) at overlapping time instants (\textit{Receiver collision}). (Receiver collision is also possible at $R_2$) This is avoided by proper scheduling of all ONUs. 

TWDM schemes of the fourth and fifth types (Fig. \ref{Fig:TWDM architectures}(d-e)) have two collision domains $C_R$ and $C_G$, at the switch and OLT receivers respectively. It is possible that multiple ONUs desiring to reach different OLT receivers (e.g., $ONU_{1,1}$ and $ONU_{2,2}$ desiring to reach $R_1$ and $R_2$ respectively), upstream simultaneously to a particular switch port ($C_G$). However, the switch can only perform one-to-one routing resulting in collision ($C_G$) at that port (Fig. \ref{Fig:Collision}). In such scenarios, although ONUs are scheduled to different receivers (i.e., $R_1$ and $R_2$), yet their upstream data ($U_{1,1}$ and $U_{2,2}$) are lost due to collision at the switch ($C_G$ in Fig. \ref{Fig:Collision}). This is termed as \textit{Group collision}. The ONUs that map to a particular switch port are assumed to form a group - $G_x$ ($ONU_{1,1}$ and $ONU_{2,2}$ belong to $G_1$ and require intra-group scheduling). 

However, even if intra-group scheduling is done perfectly, upstream from ONUs of different groups might get mapped to the same receiver at overlapping time instants (e.g., $ONU_{1,1}$ and $ONU_{2,1}$ get mapped to the same receiver $R_2$ in Fig. \ref{Fig:TWDM architectures}(d-e)) resulting in collision at $C_R$. As with the previous TWDM schemes, this is termed as \textit{Receiver collision} ($C_R$ in Fig. \ref{Fig:Collision}). Therefore, upstream data from $ONU_{1,1}$ and $ONU_{2,1}$ belonging to two separate groups $G_1$ and $G_2$ are lost due to (\textit{Receiver collision}). In order to avoid such an occurrence, proper scheduling of the ONUs to the OLT receivers is necessary. Therefore it is necessary to consider both group and receiver scheduling while designing the media access control protocol in TWDM architectures of the fourth and fifth types.

\textbf{The AWG at $RN_1$ of Fig. \ref{Fig:TWDM architectures}(d-e) can be replaced with a patch panel, as illustrated in \cite{Bhar2015}. This mitigates the physical layer problems associated with a cascaded AWG configuration.}

The TWDM scheduling algorithms proposed in literature (e.g., EFT, LFT, EFT-VF, LFT-VF - \cite{Colle2013}, \cite{Kanonakis2010}) have performed upstream scheduling of ONUs at non-overlapping time intervals (receiver scheduling). This prevents receiver collisions for the schemes of Fig. \ref{Fig:TWDM architectures}(a-c). However, for schemes illustrated in Fig. \ref{Fig:TWDM architectures}(d-e), the scheduling protocol should additionally prevent group collisions. Therefore, such schemes require a scheduling protocol that addresses both group and receiver scheduling. Absence of group scheduling in existing TWDM protocols results in throughput reduction due to group collisions. This has been illustrated in Section \ref{Section:performance} by implementing an existing TWDM scheduling scheme (EFT-VF) on a flexible TWDM network \cite{Bhar2015}, \cite{Tsalamanis2004}. In this paper we propose a protocol that considers the problem of simultaneously addressing group and receiver scheduling. We also reduce the complexity of the proposed scheme and prove that the modified scheme has linear computational complexity. A theoretical modelling of the limited bandwidth granting scheme has been developed to justify the throughput plots obtained for the proposed protocol. The rest of the paper is organised as follows; the proposed protocol is discussed in Section \ref{Section:proposedprotocol} followed by a discussion on reducing the complexity of the proposed scheme. This is followed by an analysis of the complexity of the modified scheme in Section \ref{Section:Linearise}. In section \ref{Section:performance} a comparison of performance results for existing schemes has been performed with the proposed protocol. Section \ref{Section:conclusion} concludes the paper.

\begin{figure}
	\centering
		\includegraphics[clip, page=5, trim=7cm 4cm 8cm 2cm, width=0.65\textwidth]{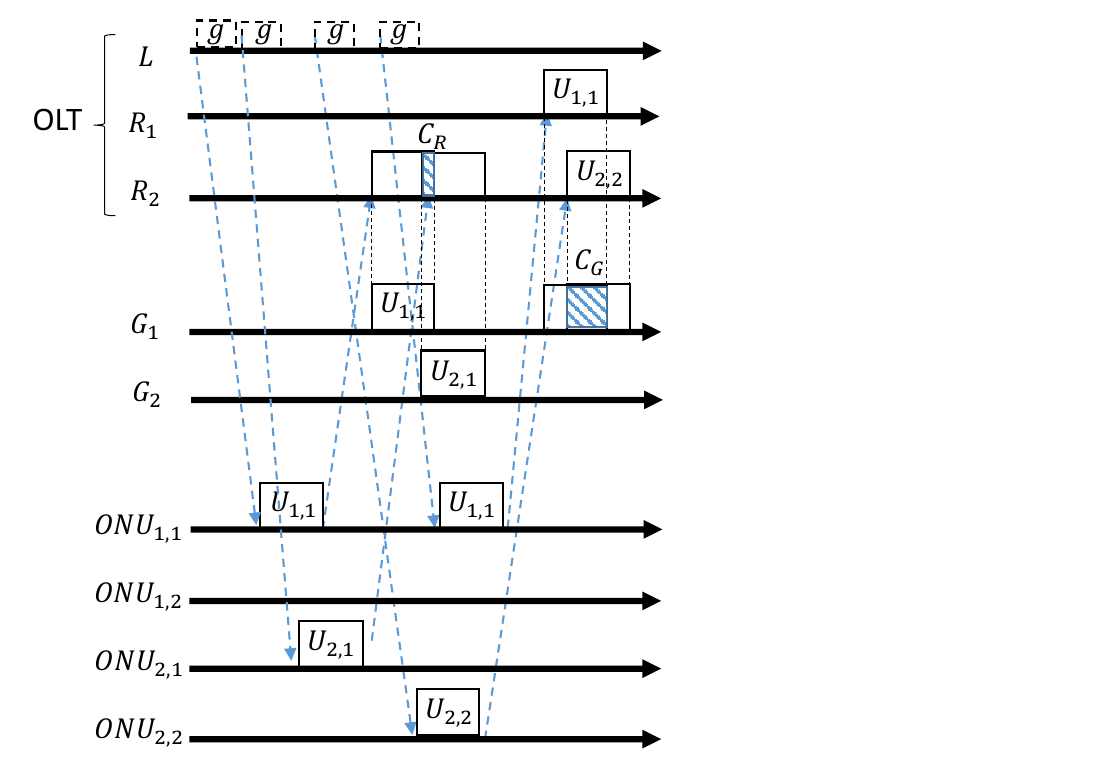}
		\caption{Illustration of the two collision domains ($C_R,~C_G$) present in flexible TWDM schemes with security. $g$ - Grant message, $U_{1,1},~U_{2,1},~U_{2,2}$ - Upstream data, $C_1,~C_2$ - Collision domains, $G_1,~G_2$ - Groups for intra-group scheduling, $R_1,~R_2$ - OLT receivers, $L$- Laser for downstream traffic at OLT.}
		\label{Fig:Collision}
\end{figure}

\section{Proposed Protocol} \label{Section:proposedprotocol}

In this section we propose a media access control (MAC) protocol for upstream data that addresses the constraints of receiver and group scheduling in a TWDM scenario. The earliest finish time with void filling (EFT-VF) \cite{Kanonakis2010} is a widely accepted TWDM scheduling scheme with excellent throughput although it may not be optimal. The MAC designed in this paper is in a close contrast with the EFT-VF scheme. The proposed MAC is referred as the constrained earliest void filling (CEVF) algorithm. 

\subsection{Protocol description - The constrained earliest void filling (CEVF) algorithm}

The CEVF algorithm uses two control messages; Request and Grant (discussed below), similar to a typical multi-point control protocol (MPCP). However, the purpose of these messages is different from a TDM based MPCP.

Request ($M_R(b)$) - This is sent by an ONU to OLT in response to a Grant message, requesting a slot to upload data. $b$ denotes the size of data in bytes that it wants to upstream to the OLT in the next cycle.

Grant ($M_G(g)$) - This is sent by the OLT to ONUs in response to a Request message. $g $ is the amount of data in bytes that an ONU is allowed to upstream.

Similar to the EFT-VF scheme, CEVF performs online scheduling of ONUs. Therefore, Grant is scheduled as soon as the Request message of an ONU is received. We assume that there are $M$ groups (equal to the number of $RN_{2,x}$), with $N$ ONUs in each group and $R$ OLT receivers. Moreover, fewer receivers might serve more groups ($R<M$), as in a typical bandwidth flexible OAN, e.g., Fig. \ref{Fig:TWDM architectures}(a) \cite{Kani2009}, \cite{Luo2013}, \cite{Effenberger2013}. Below we define two types of voids, the receiver and group voids using the illustration of Fig. \ref{Fig:CEVF}. These are essential for upstream scheduling of ONUs. An OAN with four ONUs ($ONU_{1,1},~ONU_{1,2},~ONU_{2,1},~ONU_{2,2}$) in groups of two, has been considered in Fig. \ref{Fig:CEVF}. It is assumed that the current instant is $t$ and the round trip time of $ONU_{c,d}$ ($ONU_{1,2}$ in Fig. \ref{Fig:CEVF}) is denoted by $rtt_{c,d}$ ($rtt_{1,2}$). The scheduling problem is defined as scheduling of the next $U_{c,d}$ ($U_{1,2}$), since the current $M_R(b)$ from $ONU_{c,d}$ ($ONU_{1,2}$) has been received at $t$. We illustrate that the receiver void takes care of receiver scheduling while the group void facilitates group scheduling.

\begin{figure}
	\centering
	\includegraphics[clip, page=9, trim=10.5cm 2.8cm 10.5cm 4cm, width=1\linewidth]{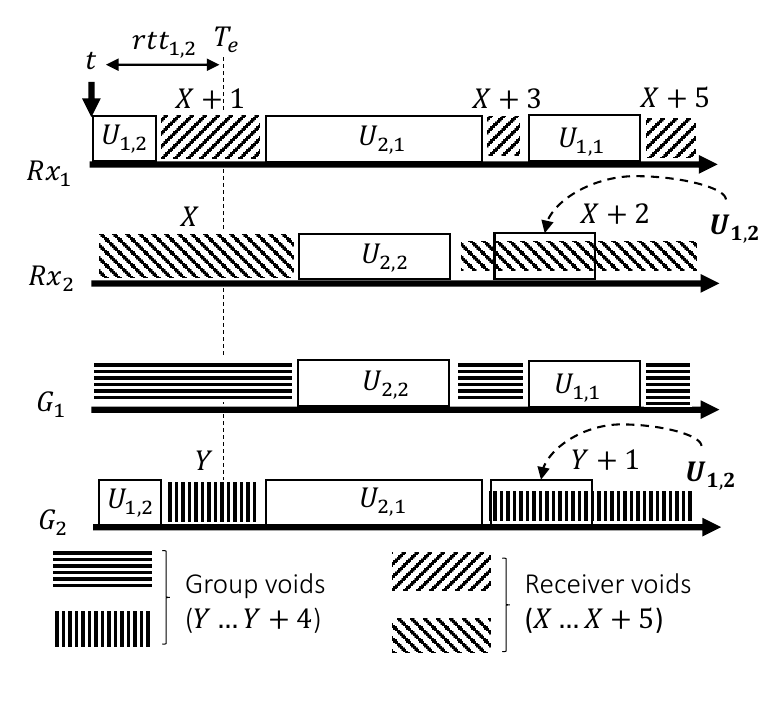}
	\caption{Illustration of receiver and group voids and CEVF working. $U_{1,1},~U_{1,2},~U_{2,1},~U_{2,2}$ - Upstream data; $G_1:~ONU_{1,1}, ~ONU_{1,2};~G_2:~ONU_{2,1}, ~ONU_{2,2}$; OLT receivers: $Rx_1,~Rx_2$; Group voids: $\{Y...Y+5\}$; Receiver voids: $\{X...X+5\}$.}
	\label{Fig:CEVF}
\end{figure}

\begin{table}
	\caption{Annotations of certain symbols used in the text}\label{tab:CAFT1} 
	\begin{tabularx}{\columnwidth}{|c|X|}
		\hline
		Symbol & Implication\\
		\hline 
		$M$ & Number of groups present\\
		$N$ & Number of ONUs present per group\\
		$R$ & Number of receivers at the OLT\\
		\bm{$V^{Rx}$} & Set of all receiver voids\\
		\bm{$V^{Rx}_r$} & Set of voids for receiver $Rx_r$\\
		$v^{Rx}$ & An element of $\bm{V^{Rx}}$\\				
		\bm{$V^G_x$} & Set of voids for group $G_x$\\
		$v^G_x$ & An element of \bm{$V^G_x$}\\
		\bm{$V^{Rx} \times V^G_x$} & The set of voids consisting of pairs $\{v^{Rx},v^G_x\}$ \\
		$S(X)$ & Starting instant of void $X$\\
		$F(X)$ & Finishing instant of void $X$\\		
		\bm{$V_{c,d}$} & Suitable voids for scheduling $ONU_{c,d}$ from \eqref{eq:start}-\eqref{eq:length}\\
		$v_{c,d}^{x,r}$ & void in which $U_{c,d}$ is scheduled\\
		$rtt_{c.d}$ & Round trip time of $ONU_{c,d}$ from the OLT\\
		$U_{c,d}$ & Upstream data for $ONU_{c,d}$ at the OLT receiver\\
		\hline
		\end{tabularx}
	\end{table}

Receiver void (\bm{$V_r^{Rx}}$;$~0\leq r \leq R$): A receiver void is the time interval during which no upstream is scheduled on a particular OLT receiver. For example, during $X$, no ONUs are scheduled to upstream on $Rx_2$ (Fig. \ref{Fig:CEVF}). Therefore, $X$ and $X+2$ are the receiver voids for $Rx_2$. The set of receiver voids for $Rx_2$ is denoted by \bm{$V_2^{Rx}$} (=$\{X,X+2\}$ in Fig. \ref{Fig:CEVF}). The start and finish times of a void $X$ are denoted by $S(X)$ and $F(X)$ respectively. The length of the void is denoted by $L(X)$ ($=F(X)~-~S(X)$). The complete set of receiver voids from all receivers is denoted by \bm{$V^{Rx}$}. Therefore \bm{$V^{Rx}=\{V_1^{Rx},V_2^{Rx},...V_R^{Rx}\}$} is a concatenation of vectors \bm{$V_r^{Rx}$}, assuming that the OLT has $R$ receivers. For every receiver $r$, there is an infinitely long void starting at $max\{max\{S(\bm{V^{Rx}_r})\},t\}$ (horizon void).

Group Void {\bm{$V_x^G$};$~0\leq x\leq M$ }: A group void is the time interval during which no upstream is scheduled from the ONUs of a particular group. Therefore $Y$ is a group void, as $U_{1,2}$ and $U_{2,1}$ are not scheduled between $S(Y)$ and $F(Y)$. The voids of a particular group $G_2$ define the set \bm{$V^G_2}$. Moreover, for every group $x$, there is an infinitely long void starting at $max\{max\{S(\bm{V^{G}_x)}\},t\}$. This is the horizon void for the respective group.

The CEVF scheme schedules any ONU; $ONU_{c,d}$ (e.g., $ONU_{1,2}$) in its respective group void \bm{$V^G_x$} (\bm{$V^G_2$}) and a suitable $Rx_r$. The group void of a particular group enforces the constraint that only one ONU from that group can upstream at any instant. This prevents group collision. Whereas receiver void enforces the condition that only one upstream is scheduled to an OLT receiver at any instant, thereby preventing receiver collision. On receiving an $M_R(b)$ from $ONU_{c,d}$, the OLT first calculates the earliest instant ($T_e$) at which $U_{c,d}$ can be scheduled. If $t$ is the current instant, then $T_e$ is defined by \eqref{eq:Te} and is illustrated in Fig. \ref{Fig:CEVF}. 

\begin{equation} \label{eq:Te}
T_e=t+rtt_{c,d}
\end{equation}

Therefore at $t$, the OLT calculates the scheduling instant for the granted $g$ bytes of $ONU_{1,2}$ ($b$ was requested in $M_R(b)$ but the OLT grants $g$; $g\leq b$). $U_{1,2}$ is scheduled to an intersection of intervals $X+2$ and $Y+1$ which is the intersection of individual receiver and group voids. In order to calculate the grant scheduling instant, we first define void intersection.

\subsubsection*{Void intersection}

Intersection of two voids ($C$), each from \bm{$V^{Rx}$} and \bm{$V_x^G$} ($A$ and $B$ respectively) is defined in \eqref{eq:intersection}. The start time of $C$ ($S(C)$) is the maximum of the starting time of the constituting elements ($A$ and $B$). Similarly, the finish time of $C$ ($F(C)$) is the minimum of the finish times of the constituting elements. For $C$ to be a valid void, $L(C)\geq 0$ must be satisfied. 

\begin{equation} \label{eq:intersection}
\begin{split}
C=A \cap B \implies S(C)=max(S(A),S(B))\\
and~F(C)=min(S(A),S(B))
\end{split}
\end{equation}

However, scheduling $U_{c,d}$ also mandates consideration of $T_e$ \eqref{eq:Te}. In order to define the set of potential voids for scheduling $U_{c,d}$ given by \bm{$V_{c,d}$}, we first assume that $v^{Rx}$ and $v_x^G$ are individual elements of \bm{$V^{Rx}$} and \bm{$V_x^G$}. Since scheduling $U_{c,d}$ requires finding an intersection of two voids, each from \bm{$V^{Rx}$} and \bm{$V_x^G$}, we define \bm{$V^{Rx} \times V_x^G$} to be the Cartesian product of \bm{$V^{Rx}$} and \bm{$V^G_x$}.  Therefore \bm{$V^{Rx} \times V_x^G$} consists of pairs $\{v^{Rx},v_x^G\}$. A potential void $v_{c,d}$ (an element of the vector \bm{$V_{c,d}$}), is an intersection of $v^{Rx}$ and $v_x^G$ as defined in \eqref{eq:intersection}. Since $v_{c,d} = v^{Rx} \cap v_x^G$ with consideration of $T_e$; $S(v_{c,d})$ is the maximum of $T_e$, $S(v^{Rx})$ and $S(v_x^G)$, \eqref{eq:start}. The finish time of $v_{c,d}$ is derived in \eqref{eq:finish} using \eqref{eq:intersection}. 

\begin{equation}\label{eq:start}
S(v_{c,d}) = max(T_e,S(v^{Rx}), S(v_x^G))
\end{equation}

\begin{equation}\label{eq:finish}	
F(v_{c,d}) = min(F(v^{Rx}), F(v_x^G))
\end{equation}

It is essential for CEVF to check that each potential void $v_{c,d}$, should be able to accommodate the $g$ bytes of $U_{c,d}$. Therefore every $v_{c,d}$ must satisfy the length criteria given by \eqref{eq:length}. The first term ($\frac{g}{l}$) in \eqref{eq:length} is the time taken to upstream $g$ bytes with link rate $l$. The receiver tuning times at the start of $U_{c,d}$ is accounted by $T_{grd}$ in \eqref{eq:length}. The receiver tuning times are considered to be included by $T_{grd}$.

\begin{equation} \label{eq:length}
\bm{V_{c,d}}=\Bigg\{v_{c,d}| L(v_{c,d}) \geq \frac{g}{l} + T_{grd}\Bigg\}	
\end{equation}

Finally, CEVF searches for the potential void with the earliest starting time. Therefore the starting time ($t_s$) for $U_{c,d}$ ($U_{1,2}$) at the OLT is given by \eqref{eq:mintime}. 

\begin{equation} \label{eq:mintime}
t_s=min(S(\bm{V_{c,d}}))
\end{equation}

Once a suitable void - $v_{c,d}^{x,r}$ (\eqref{eq:mintime}) has been found, the ONU is scheduled in that interval. The receiver to which $U_{c,d}$ is scheduled ($U_{1,2}$ is scheduled to $Rx_2$ in Fig. \ref{Fig:CEVF}) is taken from the receiver void corresponding  to $t_s$ in \eqref{eq:mintime}. The corresponding voids in \bm{$V_r^{Rx}$} and \bm{$V_x^G$} (e.g., $X+2$ and $Y+1$ in Fig. \ref{Fig:CEVF}) are updated by subtracting the upstream interval ($\frac{g}{l}+ T_{grd}$) from the respective voids. This may split the respective voids into two voids each. If any of the resultant voids are smaller than $2~\times~T_{grd}$, they are omitted from both \bm{$V_r^{Rx}$} and \bm{$V_x^G$}. 

The grant scheduling instant ($t_g$) is calculated by subtracting $rtt_{c,d}$ from $t_s$ \eqref{eq4}. It is assumed that the ONUs start upstreaming to the OLT immediately on receiving $M_G(g)$. If a limited grant allocation scheme is adopted by the OLT, then $g$ is upper bounded by some $lim$ \cite{Gravalos2014}.

\begin{equation} \label{eq4}
t_g=t_s-rtt_{c,d}
\end{equation}

\subsection{Algorithmic Reduction of the CEVF Scheme}

The CEVF algorithm finds overlapping intervals between \bm{$V^{Rx}$} and \bm{$V_x^G$} to schedule $U_{c,d}$. To achieve this, CEVF must consider one void, each from \bm{$V^{Rx}$} and \bm{$V_x^G$} and find overlapping intervals that satisfy the condition of \eqref{eq:length}. Therefore, CEVF must search for the appropriate overlapping intervals (voids) - $v_{c,d}$ within the set \bm{$V^{Rx} \times V_x^G$}. It is apparent that a grid search is necessary among all elements of \bm{$V^{Rx} \times V_x^G$} for the required void $v_{c,d}^{x,r}$. This would result in a search through $|\bm{V^{Rx}}| \times |\bm{V_x^G}|$ possibilities, where $|a|$ denotes the cardinality of the set $a$. Since $M\times N$ ONUs are scheduled on the receivers and each receiver is associated with a horizon void, $|\bm{V^{Rx}}|$ is upper limited by  $(M\times N)+R$. Moreover, $N$ ONUs are present in every group. So, $|\bm{V_x^G} | \leq N+1$. Therefore, the algorithmic complexity for CEVF would be $O((N+1)((M \times N)+R))$. In this subsection we modify the CEVF scheme to linearise the algorithmic complexity of the original CEVF. For this purpose we define ordering of voids and void hopping using the illustrations of Figs. \ref{Fig:CEVF} and \ref{Fig:orderingexample}. In Fig. \ref{Fig:orderingexample}, the timelines for $V^{Rx}_1$ and $V^{Rx}_2$ and $V^G_2$ from Fig. \ref{Fig:CEVF} have been illustrated. Although the discussion for void ordering and void hopping has been done with respect to $A$ $(\epsilon V^{Rx})$, similar void manipulations are applicable for $B$ $(\epsilon V^G_x)$.

\subsubsection{Void ordering} Ordering between two voids $A$ and $A+1$ is defined by \eqref{eq:ordering}. This implies that $A+1$ is ordered after $A$, if the starting time of $A+1$ is later than that of $A$. Therefore $X+1$ and $Y+1$ are ordered after $X$ and $Y$ as $S(X+1)>S(X)$ and $S(Y+1)>S(Y)$ respectively (Figs. \ref{Fig:CEVF}, \ref{Fig:orderingexample}).

\begin{equation} \label{eq:ordering}
 A+1 > A~iff~S(A+1)~>~S(A)~\forall~A,A+1~\epsilon~\bm{V^{Rx}}
\end{equation}

In the modified CEVF scheme (to reduce the complexity of CEVF) \bm{$V^{Rx}$} and \bm{$V^G_x$} map to \bm{$VO^{Rx}$} and \bm{$VO^G_x$}, which are ordered sets of receiver and group voids respectively. Therefore, \bm{$VO^{Rx}$} and \bm{$VO^G_x$} are given by \eqref{eq:VORX} and \eqref{eq:VOGX}.

\begin{equation} \label{eq:VORX}
\begin{split}
\bm{VO^{Rx}}=\{v_1, v_2 ...v_r\}~ where ~ S(v_1)<S(v_2)<...S(v_r)~|\\
v_1...v_r \epsilon \bm{V^{Rx}}
\end{split}
\end{equation}

\begin{equation} \label{eq:VOGX}
\begin{split}
\bm{VO^G_x}=\{v_1, ...v_g\}~ where ~ S(v_1)<S(v_2)<...S(v_g)~|\\
v_1...v_g \epsilon \bm{V^G_x}
\end{split}
\end{equation}

\subsubsection{Void hopping} The condition for hopping from one void $A$, to the next void $A+1$ is defined in \eqref{eq:Hopping}. This implies that on hopping to the next void $A+1$, no intermediate void $A'$ with $S(A)<S(A')<S(A+1)$ is skipped. However, hopping is performed within respective elements of \bm{$V^{Rx}$} or \bm{$V^G_x$} only (and not inter-set). Therefore, void hopping is performed from $X$ to $X+1$ and from $Y$ to $Y+1$ in Figs. \ref{Fig:CEVF} and \ref{Fig:orderingexample}.

\begin{equation}\label{eq:Hopping}
\begin{split}
C=A\rightarrow A+1;~~C>A,~\nexists~A'|(A'>A~\&~C>A'); \\
A,A+1\epsilon~\bm{V^{Rx}}
\end{split}
\end{equation}

\subsubsection{Algorithm}If $A~\epsilon~ V^{Rx}$ and $B~ \epsilon~ V_x^G $, then during the search process of the modified CEVF scheme, hopping is performed as; 

\begin{itemize}
	\item   $A \rightarrow A + 1~if~F(A) \leq F(B)$
	\item	$B \rightarrow B + 1~if~F(A) > F(B)$
\end{itemize}

The search is performed till a void is found that satisfies the condition \eqref{eq:length} (i.e., $v_{c,d}^{x,r}$). The algorithm terminates on finding the first such void. Therefore, instead of populating the whole set of $V_{c,d}$ and then finding $v_{c,d}^{x,r}$ from \eqref{eq:mintime} as done in the original CEVF, the modified algorithm directly finds $v_{c,d}^{x,r}$. Working principle of the modified algorithm is discussed below using the example of Fig. \ref{Fig:orderingexample}.

\begin{enumerate}
	\item $X \rightarrow X+1$ as $L(X \cap Y) <\frac{g}{l}+T_{grd}$
	\item $Y \rightarrow Y+1$ as $L(X+1 \cap Y) <\frac{g}{l}+T_{grd}$
	\item $X+1 \rightarrow X+2$ as $L(X+1 \cap Y+1) <\frac{g}{l}+T_{grd}$
	\item $v_{c,d}^{x,r}=(X+2 \cap Y+1)$ as $L(X+2 \cap Y+1)\geq \frac{g}{l}+T_{grd}$. Modified CEVF terminates.
\end{enumerate}
In the next section we prove the optimality, convergence and complexity of the modified CEVF scheme. 

%\begin{equation} \label{eq8}
%\begin{split}
%X \rightarrow~next~element~of~V^R~after~present~X~if~\\
%F(X \cap Y)\leq F(Y)
%\end{split}
%\end{equation}
%
%\begin{equation} \label{eq9}
%\begin{split}
%Y \rightarrow~next~element~of~V_x^G~after~present~Y~if~\\
%F(X \cap Y)>F(Y)
%\end{split}
%\end{equation}

\begin{figure} 
	\centering
	\includegraphics[clip, page=10, trim=10cm 8cm 10cm 4.3cm, width=0.49\textwidth]{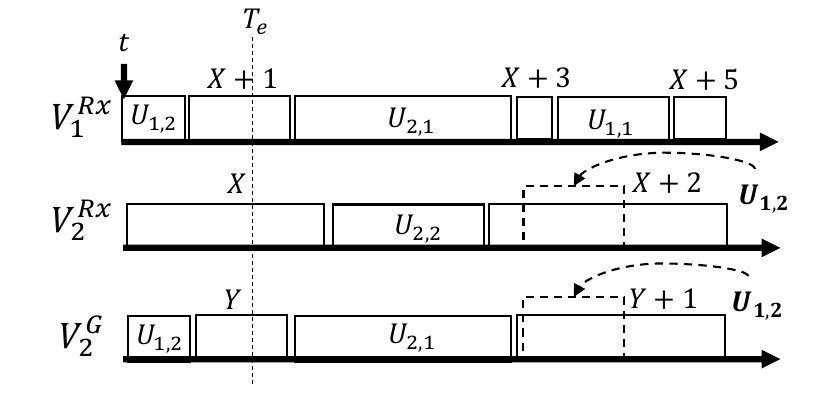}
	\caption{Illustration of void ordering and void hopping.} \label{Fig:orderingexample}
\end{figure} 

\subsection{Analysis of the Modified CEVF Scheme}\label{Section:Linearise}

The modified CEVF scheme is comprised of an online protocol which schedules ONUs on receiving the respective $M_R(b)$, and an offline protocol that maintains the ordering of voids after scheduling has been performed. Below we discuss two theorems associated with the performance of these protocols.

\subsection{Online protocol}

\begin{theorem} \label{theorem1}
	The modified CEVF algorithm requires at-most $N+(N\times M)+R$ steps to find the void $v_{c,d}^{x,r}$. Moreover, the algorithm converges, has linear complexity and is optimal.
\end{theorem}

\begin{proof}
	
For this proof we assume that $a_{k-r},a_k,a_{k+r}\epsilon \bm{V^{Rx}};~\forall ~r\epsilon I,~r>0~|a_{k+r}>a_k>a_{k-r}$ and $b_{j-g},b_j,b_{j+g} \epsilon \bm{V_x^G};~\forall~~g\epsilon I,~g>0~|b_{j+g}>b_j>b_{j-g}$. We prove the theorem for the scenario $F(a_k)<F(b_j)$, in which case $a_k \rightarrow a_{k+1}$ must be performed. We show that $\{a_k,b_{j+g}\}$ and $\{a_k,b_{j-g}\}$ pairs $\forall g > 0$ do not need to be inspected if $F(a_k)<F(b_j)$ and $L(a_k \cap b_j) <\frac{g}{l}+T_{grd}$ (i.e., $a_k$ and $b_j$ do not satisfy the length criteria \eqref{eq:length} and hence their intersection \eqref{eq:start}, \eqref{eq:finish} is not a potential void). Below we prove this two parts.

\subsubsection*{Part I} \label{Section: PartI} First we show that any $\{a_k,b_{j+g}\}$ need not be considered when performing $a_k \rightarrow a_{k+1}$. It is known that, $b_{j+g}>b_j \implies S(a_k \cap b_{j+g})~\geq~S(a_k\cap b_j)$ as $S(b_{j+g})\geq S(b_j)$. Also, $F(a_k)<F(b_j) \implies F(a_k\cap b_{j+g})\leq F(a_k\cap b_j)$. Together, these conditions imply that $L(a_k\cap b_{j+g} )\leq L(a_k\cap b_j)$. Since $L(a_k \cap b_j) <\frac{g}{l}+T_{grd}$, therefore $L(a_k\cap b_{j+g})< \frac{g}{l} + T_{grd}$. So no $a_k\cap b_{j+g}$ can satisfy the length criteria of voids given by \eqref{eq:length}. Therefore, all $(a_k,b_{j+g})$ pairs can be rejected.  

\subsubsection*{Part II}We also need to show, when performing $a_k~\rightarrow~a_{k+1}$, we do not need to consider any $\{a_k,b_{j-g}\}$. For this we consider $b_{j-g+1}$ to be the next ordered element after $b_{j-g}$ and let $a_{k-r}$ be the element, with which $b_{j-g}$ was having intersection, such that $F(a_{k-r})>F(b_{j-g})$. Therefore, we must have performed; $b_{j-g} \rightarrow b_{j-g+1}$, since $F(a_{k-r})>F(b_{j-g})$ and $L(a_{k-r}\cap b_{j-g})~<~\frac{g}{l}+T_{grd}$. From the proof of Part I, these conditions imply $L(a_k\cap b_{j-g})~<~\frac{g}{l}+T_{grd}$, as $a_k>a_{k-r}$. 

Thus each step of the modified CEVF scheme rejects multiple solutions.

A similar proof can be given to illustrate that $\{a_{k+r},b_j\}$ and $\{a_{k-r},b_j\}$ pairs $\forall r> 0$ do not need to be inspected if $F(a_k)>F(b_j)$ and $b_j~\rightarrow~b_{j+1}$ is being performed. 

Since the last voids in \bm{$V^{Rx}$} and \bm{$V_x^G$} ($max(S(\bm{V^{Rx}}))$ and $max(S(\bm{V_x^G}))$ respectively) end at infinity (horizon void), we are guaranteed to find at least one void which meets the length criteria given by \eqref{eq:length}. \textbf{Moreover, the modified CEVF selects the void with the lowest start-time}. This results in the optimality of the algorithm. In the worst case, CEVF will reach the horizon voids of both \bm{$V^{Rx}$} and \bm{$V_x^G$}. The intersection of horizon voids on \bm{$V^{Rx}$} and \bm{$V_x^G$} is infinitely long and will always meet the length criteria. In this worst case we traverse both \bm{$V^{Rx}$} and \bm{$V_x^G$} at most once. Since there are $R$ receivers, $M$ groups and $M \times N$ ONUs, the number of voids encountered while passing through \bm{$V^{Rx}$} and \bm{$V_x^G$} is $N+M\times N+R$ and hence as many steps are required in the worst case. Moreover, CEVF has a complexity of $O(N\times M)$ ($R,N<<N\times M$). Since the optimal void is found in finite number of steps, the modified CEVF algorithm is convergent.

\end{proof}

\subsection{Offline Protocol}

The offline protocol takes care of ordering the newly created voids after scheduling a particular ONU. 

\begin{theorem}
	Ordering the new void requires $\log_2(M\times N+1)$ and $\log_2(N+1)$ steps in $V^{Rx}$ and $V^G_x$ respectively.
\end{theorem}

\begin{proof}
	A binary search must be performed to order the newly created void according to \eqref{eq:ordering}. Since there are $N$ ONUs in a group, there can be a maximum of $N+1$ voids in any $V^G_x$. While a receiver $Rx_r$ can have all the ONUs scheduled to it in the worst case. As such, there will be $M\times N+1$ voids in $V^{Rx}_r$. Therefore, the binary search algorithms require $\log_2(M\times N+1)$ and $\log_2(N+1)$ steps for $V^{Rx}$ and $V^G_x$ respectively. The associated complexities are $O(\log_2 (M\times N))$ and $O(\log_2N)$ respectively.
\end{proof}

\section{Performance Results} \label{Section:performance}

In this section we compare the performance figures of the proposed CEVF algorithm with that of EFT-VF with respect to throughput figures. The simulations have been performed in OMNET++. We assume a scenario in which the OAN has 64 ONUs. Performance figures are illustrated for the scenarios in which each $RN_{2,x}$ connects to either $8$ or $4$ ($=N$) ONUs respectively. The OLT and ONU transceivers are assumed to be of 1Gbps each. The number of OLT receivers ($R$) is varied between 2, 4 and 8 to support different line rates ($r=31.25Mbps,~62.6Mbps,~125Mbps$ respectively) at the ONUs. Each ONU is assumed to be equipped with a buffer of size $1Gb$. Self-similar traffic with Pareto distributed on and off periods are generated homogeneously by the ONUs. The shape parameters for the on and off periods have been assumed to be $1.2$ and $1.4$ respectively. The maximum packet size is assumed to be limited by $1500~Bytes$ while the minimum burst size is of 1 packet.  

\begin{figure}[t!]
	\centering
	\begin{subfigure}[t]{1\linewidth}
		\centering
		\includegraphics[clip, page=11, trim=10.5cm 5.7cm 10.8cm 6cm, width=1\linewidth]{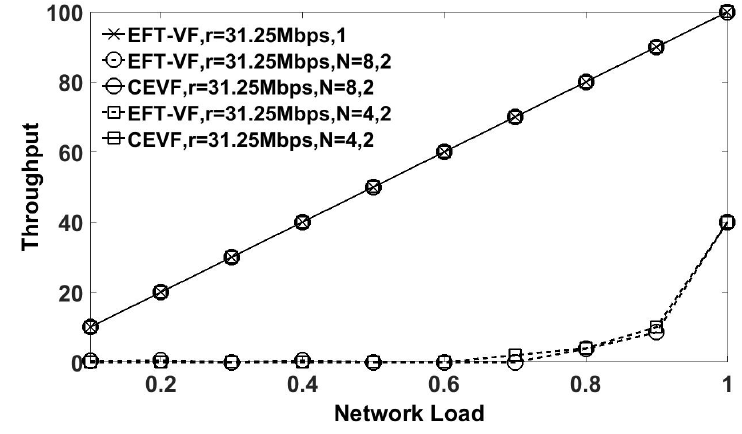}
		\caption{Throughput at $31.25Mbps$ for (1): EFT-VF applied to Fig \ref{Fig:TWDM architectures}a \cite{Luo2013} and (2): EFT-VF, CEVF applied to Fig \ref{Fig:TWDM architectures}d-e \cite{Tsalamanis2004}, \cite{Bhar2014}.} \label{Fig:throughput3125}
	\end{subfigure}	
	~
	\begin{subfigure}[t]{1\linewidth}
		\centering
		\includegraphics[clip, page=12, trim=10.5cm 5.7cm 10.8cm 6cm, width=1\linewidth]{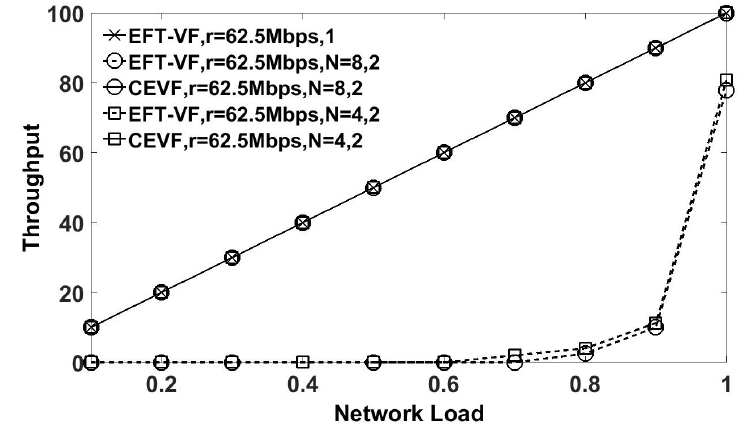}
		\caption{Throughput at $62.5Mbps$ for (1): EFT-VF applied to Fig \ref{Fig:TWDM architectures}a \cite{Luo2013} and (2): EFT-VF, CEVF applied to Fig \ref{Fig:TWDM architectures}d-e \cite{Tsalamanis2004}, \cite{Bhar2014}.} \label{Fig:throughput625}
	\end{subfigure}
	~
	\begin{subfigure}[t]{1\linewidth}
		\centering
		\includegraphics[clip, page=13, trim=10.5cm 5.7cm 10.8cm 6cm, width=1\linewidth]{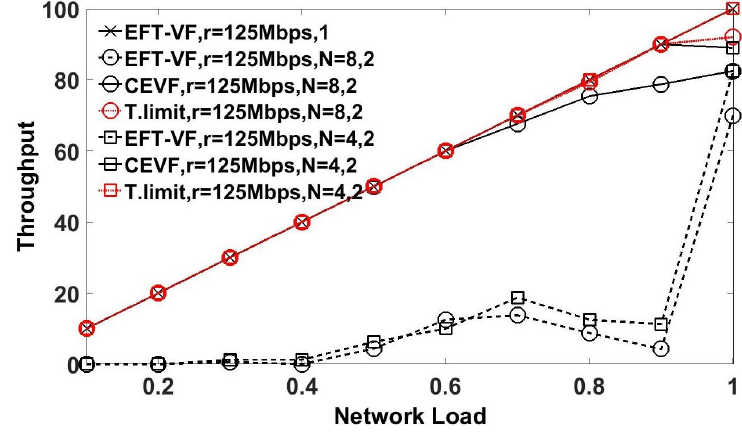}
		\caption{Throughput at $125Mbps$ for (1): EFT-VF applied to Fig \ref{Fig:TWDM architectures}a \cite{Luo2013} and (2): EFT-VF, CEVF applied to Fig \ref{Fig:TWDM architectures}d-e \cite{Tsalamanis2004}, \cite{Bhar2014}.} \label{Fig:throughput125}
	\end{subfigure}
	\caption{Comparison of the throughput performance  (as a function of the network load) between CEVF and EFT-VF for different $N$ and $r$ when applied to different TWDM schemes.}\label{Fig:throughput}
\end{figure}

\subsection*{Comparison of throughput}

A comparison of throughput has been performed in Fig. \ref{Fig:throughput} for the EFT-VF \cite{Kanonakis2010} and the CEVF schemes, assuming that these schemes are implemented either on;

\begin{enumerate}
\item the typical TWDM architecture of Fig. \ref{Fig:TWDM architectures}a or \item on flexible and secure TWDM architectures of Fig. \ref{Fig:TWDM architectures}d-e.
\end{enumerate} 
We discuss the throughput figures achieved by EFT-VF and the proposed CEVF schemes for the above scenarios.

The comparisons have been done for $r=31.25Mbps$ in Fig. \ref{Fig:throughput3125}. It is observed that CEVF achieves a higher throughput compared to EFT-VF, when applied to the flexible and secure TWDM schemes (Fig. \ref{Fig:TWDM architectures}d-e). This is attributed to the significant packet-loss associated with EFT-VF scheme due to group collisions in all scenarios $N=4,~8$. CEVF achieves considerably higher throughput figures as it considers both receiver and group scheduling. High throughput figures are also attained when EFT-VF is applied to Fig. \ref{Fig:TWDM architectures}a due to the absence of group collision domain in such schemes. The throughput for EFT-VF applied to Fig. \ref{Fig:TWDM architectures}a and CEVF schemes, at $\rho=1$ is observed to be $\sim 99\%$. An increase in the line-rate to $62.5Mbps$ (Fig. \ref{Fig:throughput625}) results in similar throughput performance as in Fig. \ref{Fig:throughput3125}. 

At $r=125Mbps$, the throughput of CEVF saturates approximately at eighty percent and eighty-five percent for for $N=8$ and $N=4$ respectively. This is attributed to two reasons, that are discussed below.

\begin{enumerate}
	\item The additional constraint of group collision in CEVF, results in formation of unused voids. The burstiness of the Pareto distributed on-off traffic sources also increases with network load. This makes it difficult for CEVF to find appropriate voids. Moreover, the chances of group collision increase with an increase in $N$, thereby making group scheduling more difficult.
	
	\item As discussed earlier, presence of the group collision domain imposes the condition that the $N$ ONUs of a particular group communicate on a particular feeder fiber and switch port at the OLT. Since one-to-many port mapping is not possible at the switch port, the ONUs of any group can be allocated the effective line rate of one OLT receiver only. This limits the line rate ($r$) of ONUs in Fig. \ref{Fig:TWDM architectures}d-e according to: $N\times r \leq 1Gbps$ (assumed data-rate of OLT receivers). Moreover, for scenarios with high traffic load and $N\times r = 1Gbps$, the throughput declines, as observed from Fig. \ref{Fig:throughput125}. This is because, the limited granting scheme is associated with an upper limit of throughput at high network loads (when $N\times r = 1Gbps$). This phenomenon is explained below using a theoretical model.
\end{enumerate}
 
We derive the upper-bound for the throughput in a limited granting scheme by constructing a Markov chain. The bandwidth granted to an ONU is assumed to be upper limited by $lim$. For the architectures of Fig. \ref{Fig:TWDM architectures}d-e, at high traffic-load conditions and $N \times r=1Gbps$, the $N$ ONUs of a particular group will share the bandwidth of one OLT receiver only, as discussed earlier. This corresponds to a time division multiplexed scenario with the limited granting scheme (upper limit of bandwidth granted - $lim$) \cite{Gravalos2014}, over the effective bandwidth of one OLT receiver. It also allows us to define the concept of time cycle. Therefore, the inter-scheduling duration of every ONU is upper-limited by $lim \times N$ (length of time-cycle). The state of the ONU buffer at the end of every time cycle ($B_i$ at the end of $i^{th}$ time cycle) is assumed to be the state variable. Furthermore, for a network load of $\rho$, the length of the time cycle is assumed to be of $T= \rho \times lim \times N$ duration ($lim \times N$ is assumed to be $2ms$). A discrete time Markov chain is formulated for this scenario. The state transitions correspond to the change in buffer-state of an ONU, between the end of the previous and the present time cycles. 

The following conditions can occur for the state variable at the end of two consecutive time cycles ($i-1$ and $i$ respectively). The corresponding state transition probabilities are also mentioned alongside.

\begin{enumerate}
	\item $B_{i}=0$ and $0 \leq B_{i-1} \leq lim$. State transition probability: $\sum_{k=0}^{lim-B_{i-1}} \frac{e^{-\lambda T} (\lambda T)^k}{k!}$.
	\item $0<B_{i}\leq B_{i-1}$ and $B_{i-1} \leq B_{i} + lim$. State transition probability: $\frac{e^{-\lambda T}(\lambda T)^{(B_i+lim-B_{i-1})}}{(B_i+lim-B_{i-1})!}$.
	\item $B_{i-1}<B_i$. State transition probability: $\frac{e^{-\lambda T}(\lambda T)^{(B_i-B_{i-1})}}{(B_i-B_{i-1})!}$
\end{enumerate}

The state transition probability matrix $P$ is formulated with the above mentioned probabilities. The Markov chain has finite number of states as the buffer capacity of each ONU is fixed at $1Gb$. Moreover, the Markov chain is connected and the state transition probabilities are independent of time, resulting in the time-homogeneous nature of the Markov chain. This results in the ergodic nature of the Markov chain allowing us to solve for the steady state probabilities $\pi(B_i)$, using equations $\pi(B_i)=\pi(B_i)\times P$ and $\sum_{B_i=0}^{\infty} \pi (B_i)=1$. The Markov chain has been solved in MATLAB.

The loss in throughput is the probability of buffer overflow. Therefore, throughput \eqref{eq:thru} is given by the probability that buffer does not overflow ($B_\infty$ corresponds to the scenario that buffer is full, i.e. $1Gb$). The throughput obtained from \eqref{eq:thru} is plotted in Fig. \ref{Fig:throughput125}, as the theoretical upper-bound of a limited granting scheme for $r=125Mbps$ and $N=8,4$. The theoretical bound on throughput has been obtained for each ONU buffer, while the plot obtained from the simulation is for throughput of the entire network Fig. \ref{Fig:throughput125}. Since, ONUs have been assumed to be homogeneous, both plots eventually correspond. The theoretical bound explains the drop in throughput for CEVF. The difference between the theoretical bound and the plot obtained from simulation in Fig. \ref{Fig:throughput125} arises due to the assumption of exponential traffic in theoretical modelling. However, the Pareto distributed on-off traffic exhibits bursty nature, resulting in saturation of throughput at much lower network loads $\rho \sim 0.8$.

\begin{equation}\label{eq:thru}
throughput=\rho \times (1- Pr(B_\infty)) \times
100
\end{equation}

For EFT-VF applied to Fig. \ref{Fig:TWDM architectures}d-e, the throughput initially increases with an increase in network load ($\sim \rho=0.4$) and falls thereafter ($\sim \rho=0.9$). This nature of corresponds to the throughput of the slotted-ALOHA scheme. The throughput finally increases at $\rho~0.9$. This is because at high network loads, CEVF reduces to a TDM scheme, thereby reducing group collisions.

\section{Conclusion}\label{Section:conclusion}

In this paper we have identified a new scheduling domain (group collision) that is present in flexible and secure TWDM schemes. A theoretical upper bound of the throughput for these schemes has also been derived. We have illustrated that the proposed CEVF scheduling scheme achieves a significantly high throughput as it considers both group and receiver collisions. The obtained throughput is closely matches with the theoretical bound. A modified version of the CEVF has also been proposed in this paper, that is an optimal scheme and has linear complexity. This results in low computational requirements.

% or


%\appendices

%\section{}

%\section*{Acknowledgment}

\ifCLASSOPTIONcaptionsoff
  \newpage
\fi

\bibliographystyle{IEEEtran}


\end{document}